\documentclass[nofootinbib,twocolumn,pra,letterpaper]
{revtex4-2}
\usepackage{hyperref}
\usepackage{latexsym,amsmath,amssymb,amsfonts,graphicx,color,amsthm}
\usepackage{enumerate}
\usepackage{braket}
\usepackage{adjustbox}
\usepackage{footnote}
\usepackage[dvipsnames]{xcolor}

\newcommand{\cA}{\mathcal{A}}
\newcommand{\cB}{\mathcal{B}}

\newcommand{\cE}{\mathcal{E}}

\newcommand{\cH}{\mathcal{H}}

\newcommand{\cM}{\mathcal{M}}
\newcommand{\cN}{\mathcal{N}}

\newcommand{\cP}{\mathcal{P}}
\newcommand{\cQ}{\mathcal{Q}}

\newcommand{\cS}{\mathcal{S}}

\newcommand{\cX}{\mathcal{X}}

\newcommand{\Id}{\mathbb{I}}
\newcommand{\tr}{\text{Tr}}

\newtheorem{theorem}{Theorem}

\newtheorem{lemma}[theorem]{Lemma}

\newtheorem{example}{Example}

\newtheorem{remark}{Remark}

\begin{document}

\title{Characterizing incompatibility of quantum measurements via their Naimark extensions}
\author{Arindam Mitra$^{1,2}$}
\affiliation{$^1$Optics and Quantum Information Group, The Institute of Mathematical Sciences,
C. I. T. Campus, Taramani, Chennai 600113, India.\\
$^2$Homi Bhabha National Institute, Training School Complex, Anushakti Nagar, Mumbai 400094, India.}
\author{Sibasish Ghosh$^{1,2}$}
\affiliation{$^1$Optics and Quantum Information Group, The Institute of Mathematical Sciences,
C. I. T. Campus, Taramani, Chennai 600113, India.\\
$^2$Homi Bhabha National Institute, Training School Complex, Anushakti Nagar, Mumbai 400094, India.}
\author{Prabha Mandayam}
\affiliation{Department of Physics, Indian Institute of Technology Madras, Chennai - 600036, India.}

\date{\today}

\begin{abstract}
We obtain a formal characterization of the compatibility or otherwise of  a set of positive-operator-valued measures (POVMs) via their Naimark extensions. We show that a set of POVMs is jointly measurable if and only if there exists a single Naimark extension, specified by a fixed ancilla state on the same ancilla Hilbert space, that maps them to a set of commuting projective measurements (PVMs).  We use our result to obtain an easily checkable sufficient condition for the compatibility of a pair of dichotomic observables in any dimension. This in turn leads to a characterization of the compatibility regions for some important classes of observables including a pair of unsharp qubit observables. Finally, we also outline as to how our result provides an alternate approach to quantifying the incompatibility of a general set of quantum measurements. 
\end{abstract}

\maketitle

The existence of incompatible measurements is a fundamental feature of quantum theory~\cite{heisenberg} which has been well-studied, both from the point of view of quantum foundations~\cite{heinosaari} as well as quantum information processing~\cite{mubs, qkd_d,Heino_ran}. Of central interest is the question of characterizing the incompatibility of a set of quantum measurements. It is well known that the incompatibility of a pair of projection-valued measurements (PVMs) is synonymous with their non-commutativity~\cite{lahiti}. Formally, a pair of PVMs is jointly measurable if and only if they pairwise commute~\cite{lahiti}, that is, if and only if each effect of one PVM commutes with all other effect of any other PVM. Furthermore, for a set of PVMs, it is known that the existence of a joint measurement for the whole set is equivalent to the existence of joint measurements for every pair of PVMs in the set~\cite{heinosaari08, rkunj_2014, Oh_1, Oh_2}.

Motivated by this connection between joint measurability and commutativity, several measures of incompatibility have been proposed and studied for a set of PVMs~\cite{wehner_winter, incompatibility_BM, PM_MS1}. While the  measures discussed in~\cite{wehner_winter, incompatibility_BM} focus on rank-one PVMs, the incompatibility measure in~\cite{PM_MS1} holds for general PVMs with higher-rank projectors. 

The equivalence between joint measurability and commutativity breaks down in the case of general quantum measurements, beyond projective measurements. This is best captured by the well known example of a set of three positve-operator-valued measures (POVMs) in two dimensions which are pairwise jointly measurable, but for which there does not exist a triplewise joint measurement~\cite{heinosaari08, jointmeas_graph, Liang, Kunjwal_Sibasish}. It has been proposed to quantify the incompatibility of a set of POVMs based on the closest distance to a set of compatible measurements, or, alternately, using the amount of noise required to approximate an incompatible set using a set of compatible measurements~\cite{heinosaari}. 

In this work, we propose an alternate approach to quantifying the incompatibility of a set of POVMs, using the Naimark extension. Recall that Naimark's theorem states that every POVMS can be realised as a PVM in an extended Hilbert space. Previously, the Naimark extension has been used to obtain a necessary and sufficient condition for a pair of POVMs to be incompatible~\cite{beneduci,kiukas}. More recently, the Naimark extension has been used to obtain a resource theory of coherence for POVMs~\cite{coherence_povm}. 

Here we show that if a given set of POVMs (acting on the states of a quantum system) are jointly measurable, one can always find out a corresponding set of \emph{compatible} Naimark extensions of these POVMs, in one and the same Hilbert space. Conversely, we argue that if the set of Naimark extensions corresponding to a given set of POVMs are compatible, it does not necessarily imply that the POVMs are jointly measureable. In the special case where the said Naimark extensions use one and the same ancilla state, the corresponding POVMs can be jointly measured.  

The rest of the paper is organised as follows. In Sec.~\ref{sec:prelim} we briefly review the mathematical formalism of the Naimark dilation theorem. We then recall the known results relating to incompatibility of a pair of POVMs in Sec.~\ref{sec:2povms}. In Sec.~\ref{sec:Npovms}, we present our result on the relation between the incompatibility of a set of $N$ POVMs and the associated PVMs constructed via Naimark extension. We discuss two specific applications of our result in Sec.~\ref{sec:applications}. In Sec.~\ref{sec:measure}, we present an alternate approach to quantify incompatibility of general quantum measurements using the Naimark extension. In Sec.~\ref{sec:dichotomic}, we show that our result can be used to obtain a sufficient condition for the compatibility of a pair of dichotomic observables. This in turn leads to a characterization of the compatibility of a pair of unsharp qubit observables, as shown in Sec.~\ref{sec:unsharp}. Finally, we conclude with a summary and future outlook in Sec.~\ref{sec:concl}. 

\section{Preliminaries}\label{sec:prelim}

A projection-valued measure (PVM) $\cP$ (von Neumann measurement, also called a `sharp measurement') with $n$ outcomes on a $d$-dimensional Hilbert space $\cH_{S}$ is characterised by a set of $n$ projection operators  $\{P (i), \; i=1,\ldots,n\}$, such that $\sum_{i=1}^{n}P(i) = I_{\cH_{S}}$. The probability of obtaining outcome $i$, when measuring state $\rho \in \cB(\cH_{S})$ is given by $p(i) = \tr[\rho P (i)]$. Here and in what follows, $\cB(\cH_{S})$ denotes the set of bounded linear operators acting on the Hilbert space $\cH_{S}$. 

A positive-operator-valued measure (POVM) $\cM$ with $n$ outcomes on a $d$-dimensional Hilbert space space $\cH_{S}$ is described by a set of $n$ positive semi definite operators, called \emph{effects}, $\{E (j), j = 1,2,\ldots, n\}$, satisfying ,
\[I_{\cH_{S}} \geq E(j) \geq 0, \forall j; \; \sum_{j=1}^{n}E (j) = I_{\cH_{S}}. \]
For any state $\rho \in \cS(\cH_S)$, the probability of obtaining outcome $j$ is given by $p(j) = tr[\rho E(j)]$, where $\cS(\cH_S)$ denotes the state space of the system.  For projective measurements, there is a canonical description of the post-measurement state corresponding to a given outcome ($i$) in terms of the corresponding projection operator, via the map $\rho \rightarrow \frac{P(i)\rho P(i)}{\tr(\rho P(i))}$. There is no such canonical association for POVMs. For example, one set of measurement operators that maybe associated with the POVM $\cM=\{E(j)\}$ is the set $\{A (j)\}$ satisfying $E (j) = A (j)^{\dagger}A(j)$, so that the post-measurement state is obtained as $\rho \rightarrow \frac{A(j)\rho A^{\dagger} (j)}{\tr[A(j)\rho A^{\dagger} (j)]}$. This implementation of the POVM $\cM$ is often referred to as the \emph{L\"uders instrument}~\cite{Heinosaari_Ziman}. 

Two POVMs $\cM_{1}=\{E_{1}(j_{1})\vert j_{1} = 1,\ldots,n_{1}\}$ and $\cM_{2}=\{E_{2}(j_{2})\vert j = 1,\ldots,n_{2}\}$ are said to be jointly measurable if there exists a \emph{joint} measurement $\cM=\{ E(i,j)\}$, such that the POVMs $\cM_{1}$ and $\cM_{2}$ can be realised as the marginals of the POVM $\cM$. In other words, the elements $\{E_{i}(j_{1})\}$ and $\{E_{2}(j_{2})\}$ can be realised as,
\[ E_{1}(j_{1}) = \sum_{j_{2}=1}^{n_{2}} E(j_{1},j_{2}); \; E_{2}(j_{2})= \sum_{j_{1}=1}^{n_{1}} E(j_{1},j_{2}),  \forall j_{1},j_{2}.\]

Naimark's theorem~\cite{Peres, paris_Naimark} states that every POVM on a $d$ dimensional Hilbert space $\cH_{S}$ can be realised as a projection-valued measure (PVM) on an extended Hilbert space $\cH'$ of dimension $d'> d$. Such an extension may be obtained using a direct sum or a direct product construction. Specifically, let $\cP = \{P(i)\}$ denote the PVM acting on the space $\cH'$, obtained via a Naimark extension of the POVM $\cM = \{M (i)\}$. One approach to construct the elements of the PVM $\cP$ is by using a direct sum extension of the space $\cH_{S}$ to the space $\cH' \equiv \cH_{S} \oplus \cH_{A}$, such that,
\[ \tr[ E (i) \rho] = \tr \left[ P(i)(\rho \oplus \mathbf{0} )\right], \forall i, \; \forall \rho \in \cB(\cH_S), \]
where $\mathbf{0}$ is the null matrix of dimension $(d_{A}-d)$. 

Alternately, we could use the so-called \emph{canonical} Naimark extension, where the POVM is realised via a projective measurement on an ancilla system, after the system and the ancilla interact via a suitable unitary. Let $\cH_{A}$ denote the Hilbert space of the ancilla system and suppose the ancilla starts out in the state $\sigma_{A} \in \cB(\cH_{A})$, where $\cB(\cH_{A})$ denotes the set of bounded linear operators on the Hilbert space $\cH_{A}$. Then, by Naimark dilation theorem, there exists a PVM $\cP=\{P(i)\}$ acting on $\cH\otimes \cH_{A}$, such that,  
\begin{equation}
 \tr_{\cH_{S}} [\rho M (i)] = \tr_{\cH_{S}\otimes \cH_{A}}\left[P(i) (\rho \otimes \sigma_{A}) \right],
 \end{equation}
 for all states $\rho \in \cB(\cH_{S})$.
As evident in the case of direct sum based Nairmark extenstion the smallest dimension of the ancilla systems needed for such a tensor product Naimark extension is known to be $d_{A} = \sum_{i=1}^{n} r_{i} - d$, where $r_{i}$ is the rank of the $i^{\rm th}$ POVM effect~\cite{ancilla_dim}. 

The Naimark extension is not unique and the PVM associated with a specific POVM depends on the dimensions of the ancilla system $\cH_{A}$ as well as the fiducial state that the ancilla systems starts out in. In this note, we use the canonical, tensor-product form of the Naimark extension in our proofs. However, our results can equivalently be proved using the direct sum form as well.

\subsection{Compatibility of quantum measurements} \label{sec:2povms}

We next review some of the known results on the compatibility of a set of quantum measurements. First, we note the well known result on the joint measurability of a set of PVMs, originally proved in~\cite{heinosaari08, rkunj_2014}. 


\begin{lemma}[Compatibility of PVMs]
A set of PVMs $\Pi = \{\cP_1,\cP_2, \ldots, \cP_{N}\}$ is compatible, iff the PVMs $\{\cP_{j}\}$ in that set  are pairwise commuting. Furthermore, there exists a unique joint measurement which is a projective measurment. \label{lem:pvm}
\end{lemma}

Such a necessary and sufficient condition for a pair of PVMs ($N=2$)~\cite{beneduci, lahiti} in turn implies the following necessary and sufficient condition for a pair of POVMs to be compatible, originally proved in~\cite{beneduci, kiukas,Haapasalo} using the Naimark dilation theorem. We rewrite the proof here, in terms of the tensor product Naimark extension, in preparation for our main result on the compatibility of a set of $N$ POVMs. 

\begin{theorem}
Two POVMs $\cM_{1}, \cM_{2}$ are compatible if and only if there exists at least one Naimark extension of both POVMs to a pair of commuting PVMs. \label{thm:2povms}
\end{theorem}

\begin{proof}
Consider a pair of POVMs $\cM_1=\{E_{1}(j_{1})\, \vert \, j_{1} \in [1, n_{1}] \}$ and $\cM_{2}=\{E_{2}(j_{2}) \, \vert \, j_{2} \in [1,n_{2}]\}$ that are compatible, where we use the notation $[1,n_{i}]$ to denote the set $\{1,2,\ldots,n_{i}\}$. By definition, there exists a joint POVM $\cM=\{E (j_{1},j_{2})\}$ such that $\sum_{j_{2}}E(j_{1},j_{2})=E_{1}(j_{2})$ and $\sum_{j_{1}}E(j_{1},j_{2}) = E_{2}(j_{1})$, for all $j_{1}, j_{2}$. Let $\cP=\{P (j_{1},j_{2})\}$ be the projection-valued operation obtained via a Naimark extension of $\cM$ using the fiducial ancilla state $\sigma_{A} \in \cB(\cH_{A})$, so that, for all $j_{1} \in [1,n_{1}]$, $j_{2} \in [1,n_{2}]$,
\begin{equation}
\tr_{\cH_{S}\otimes \cH_{A}}\left[ P(j_{1},j_{2}) (\rho\otimes \sigma_{A}) \right] =\tr_{\cH_{S}} [ \rho E(j_{1},j_{2}) ].
\end{equation}
Now we note that,
\begin{eqnarray}
&& \tr_{\cH_{S}\otimes \cH_{A}} \left[ \sum_{j_{2}=1}^{n_{2}} P(j_{1},j_{2}) (\rho\otimes \sigma_{A}) \right] \nonumber \\
&=& \tr_{\cH_{S}} \left[ \rho \sum_{j_{2}=1}^{n_{2}}  E( j_{1},j_{2} ) \right]  = \tr_{\cH_{S}} [\rho E_{1}(j_{1}) ]. \nonumber
\end{eqnarray}
Thus, $\cP_{1} = \{ P_{1}(j_{1}) = \sum_{j_{2}=1}^{n_{2}} P (j_{1},j_{2}) \}$ is the PVM obtained using a Naimark extension of $\cM_{1}$, with the same fiducial ancilla state $\sigma_{A} \in \cB(\cH_{A})$. Similarly, $\cP_2 = \{P_{2}(j_{2}) = \sum_{j_{1}=1}^{n_{1}} P(j_{1},j_{2})\}$ is PVM obtained via a Naimark extension of $\cM_{2}$, using the same ancilla state $\sigma_{A} \in \cB(\cH_{A})$, as shown below.
\begin{eqnarray}
&& \tr_{\cH_{S}\otimes \cH_{A}}\left[ \sum_{j_{1}=1}^{n_{1}} P ( j_{1}, j_{2} ) (\rho\otimes\sigma_{A}) \right] \nonumber \\
&=& \tr_{\cH_{S}} [ \rho\sum_{j_{1}=1}^{n_{1}} E (j_{1},j_{2}) ] =\tr_{\cH_{S}}[ \rho E_{2}(j_{2}) ]. \nonumber
\end{eqnarray}
Since $\cP_2 = \{ P_{2}(j_{2}) = \sum_{j_{1}}P_(j_{2},j_{2}) \}$ and $\cP_{1} =\{ P_{1}(j_{1}) = \sum_{j_{2}} P (j_{1}, j_{2})\}$, the two PVMs $\cP_{1}$ and $\cP_{2}$ are compatible, via the joint observable $\cP=\{P (j_{1},j_{2})\}$. Hence, by Lemma~\ref{lem:pvm} they commute.

We now prove the converse. Suppose two POVMs $\cM_1= \{E_{1} (j_{1}) \}$ and $\cM_{2} = \{E_{2}(j_{2})\}$ can be extended to a pair of commuting PVMs $\cP_{1} = \{P_{1}(j_{1})\}$ and $\cP_{2} = \{P_{2}(j_{2})\}$ respectively, via Naimark extensions using the same ancilla state $ \sigma_{A} \in \cB(\cH_{A})$. By Lemma~\ref{lem:pvm}, there exists a PVM $\cP=\{P (j_{1},j_{2})\}$ such that $\cP_{2} \sim \{P_{2}(j_{2})=\sum_{j_{1}} P(j_{1},j_{2})\}$ and $\cP_{1} \sim \{P_{1}(j_{1})=\sum_{j_{2}}P (j_{1},j_{2}) \}$. Now, we construct the POVM $\cM$ whose effects are obtained as
\begin{equation}
\tr_{\cH_{S}\otimes\cH_{A}}\left[ (\rho\otimes\sigma_{A}) P(j_{1},j_{2}) \right]=\tr_{\cH_{S}}[ \rho M (j_{1},j_{2}) ]. 
\end{equation}
It is easy to check that $\cM$ is indeed the joint measurement for the POVMs $\cM_{1}$ and $\cM_{2}$. Specifically,
\begin{eqnarray}
\tr_{\cH_{S}}[ \rho E_{1}(j_{1})] &=& \tr_{\cH_{S}\otimes \cH_{A}} [ (\rho\otimes \sigma_{A}) P_{1}(j_{1}) ]\nonumber\\
&=& \tr_{\cH_{S}\otimes \cH_{A}}[ (\rho\otimes\sigma_{A}) \sum_{j_{2}}P(j_{1},j_{2}) ] \nonumber \\
&=& \sum_{j_{2}}\tr_{\cH_{S}\otimes \cH_{A}} \left[ (\rho\otimes\sigma_{A}) P(j_{1},j_{2}) \right] \nonumber \\
&=& \sum_{j_{2}} \tr_{\cH_{S}}[ \rho E(j_{1},j_{2}) ], \nonumber
\end{eqnarray}
and,
\begin{eqnarray}
\tr_{\cH_{S}}[ \rho E_{2}(j_{2}) ] &=& \tr_{\cH_{S}\otimes \cH_{A}}[ (\rho\otimes\sigma_{A}) P_{2}(j_{2})]\nonumber\\
&=& \tr_{\cH_{S}\otimes \cH_{A}}[ (\rho\otimes\sigma_{A})  \sum_{j_{1}} P(j_{1},j_{2}) ] \nonumber \\
&=& \sum_{j_{1}}\tr_{\cH_{S}\otimes \cH_{A}}[ (\rho\otimes\sigma_{A}) P(j_{1},j_{2}) ] \nonumber \\
&=& \sum_{j_{1}} \tr_{\cH_{S}}[ \rho E(j_{1},j_{2}) ] . \nonumber
\end{eqnarray}
Hence, $\cM_{1} =\{E_{1}(j_{1})\}$ and $\cM_{2}=\{E_{2}(j_{2})\}$ are compatible, with the joint POVM given by $\cM=\{E(j_{1},j_{2}) \}$.
\end{proof}

We note here that one crucial step in the necessary and
sufficient condition for the compatibility of a pair of
POVMs is the existence of a single Naimark extension -- characterised by a specific ancilla space $\cH_{A}$ and a fixed
ancilla state $\sigma_{A}$ -- that extends both the POVMs $\cM_{1}$ and $\cM_{2}$ to a pair of commuting PVMs. 

\section{Incompatibility of a set of $N \geq 2$ quantum measurements}\label{sec:Npovms}

We are now ready to state and prove our first result on the incompatibility of a set of POVMs. Specifically, we consider a set of $N$ POVMs, denoted as $\chi = \{\cM_{i} \, \vert\,  i=1, 2, \ldots, N\}$, acting on a Hilbert space $\cH_{S}$. Let 
\[ \{E_i(j_i) \, \vert \, I_{\cH_{S}} \geq E_i(j_i)\geq 0 \} ,\]
denote the effects corresponding to the POVM $\cM_{i}$, $\forall i \in \{1,\ldots,N\}$, such that,
\[ \sum_{j_i =1}^{n_{1}} E_i (j_i) = I_{\cH_{S}}, E_i(j_i)\in \cB(\cH_{S})~\forall j_i \in [1,n_{i}]. \] 

\begin{theorem}\label{thm:N_povm}
The set $\chi = \{\cM_{i} \, \vert \, i \in [1,N]\}$ is compatible iff there exists a single Naimark extension -- characterised by a specific ancilla Hilbert space $\cH_{A}$ and a fixed ancilla state $\sigma_{A}$ -- by means of which the POVMs $\{\cM_{i}\}$ in the set $\chi$ can be realised as a set of pairwise commuting PVMs on the extended Hilbert space $\cH_{S} \otimes \cH_{A}$. 
\end{theorem}

\begin{proof}

We first show the sufficient condition that if the set $\chi$ of $N$ POVMs $\{\cM_{i} \vert \, i=1,2,\ldots, N\}$ is compatible, it can be realised as a set of pairwise commuting PVMs on an extended Hilbert space. 
By definition, there exists a joint POVM $\cM$ with effects 
\[ \{ E(j_{1}, j_{2}, \ldots, j_{N}) \, \vert \, I_{\cH_{S}} \geq E(j_{1}, j_{2}, \ldots, j_{N}) \geq 0 \}, \] 
satisfying $\sum_{j_{1}, j_{2}, \ldots j_{N}} E(j_{1}, j_{2}, \ldots, j_{N}) = I_{\cH_{S}}$. The effects of the individual POVMs $\cM_{i}$ are indeed obtained as the marginals of the effects of the joint measurement as follows.
\begin{eqnarray}
E_{i}(j_{i}) &=& \sum_{j_{1}=1}^{n_{1}}\ldots\sum_{j_{i-1}=1}^{n_{i-1}}\sum_{j_{i+1}=1}^{n_{i+1}}\ldots\sum_{j_{N}=1}^{n_{N}} E(j_{1}, \ldots, j_{N}) , \nonumber \\
&& \; \; \forall \; j_{i} \in \{1,2,\ldots,n_{i}\}, \; i \in \{1,2,\ldots,N\}.   \label{eq:povm_marginal}
\end{eqnarray}

By Naimark dilation theorem, we can find a projective measurement $\cP $ acting on the extended space $ \cH_{S}\otimes \cH_{A} $, characterized by the set of projectors $\{ P(j_{1}, j_{2}, \ldots, j_{N})\}$. For all $j_{i} \in [1, n_{i}]$, the elements of the PVM $\cP$ are related to the effects of the POVM $\cM$ via a fixed ancilla state $\sigma_{A} \in \cB(\cH_{A})$, as,
\begin{eqnarray}
&& \tr_{\cH_{S}\otimes \cH_{A}} [ P (j_{1}, j_{2}, \ldots, j_{N}) (\rho\otimes\sigma_{A})]  \nonumber \\
&=&  \tr_{\cH_{S}} [ E(j_{1}, j_{2}, \ldots, j_{N}) \rho ], \; \; \forall  \, \rho \in \cB(\cH_{S}). \label{eq:naimark}
\end{eqnarray}
Note that the projective elements of the measurement $\cP$ satisfy $\sum_{j_{1}, j_{2}, \ldots, j_{N}} P(j_{1}, j_{2}, \ldots, j_{N}) = I_{\cH_{S}\otimes\cH_{A}}$ and $P(j_{1}, j_{2}, \ldots, j_{N})P(k_{1}, k_{2}, \ldots, k_{N}) = \delta_{j_{1}k_{1}}\delta_{j_{2}k_{2}}\ldots\delta_{j_{N}k_{N}}P(j_{1}, j_{2}, \ldots, j_{N})$
Now we define operators $\{P_{i}(j_{i)}\}$ to be the marginals of the projectors $P (j_{1}, j_{2}, \ldots, j_{N})$, obtained as,
\begin{eqnarray}
&& P_{i}(j_{i}) \nonumber \\
&=&  \sum_{j_{1}=1}^{n_{1}}\ldots\sum_{j_{i-1}=1}^{n_{i-1}}\sum_{j_{i+1}=1}^{n_{i+1}}\ldots\sum_{j_{N}=1}^{n_{N}} P (j_{1}, j_{2}, \ldots, j_{N}),  \nonumber \\
&& \quad \quad \forall j_{i} \; \in \{1,2,\ldots,n_{i}\}, \; i \in \{1,2,\ldots,N\}. \label{eq:pvm1}
\end{eqnarray}
We first note that, 
\begin{equation}
\sum_{j_{i}=1}^{n_{i}} P_{i}(j_{i})  = \sum_{j_{1}, j_{2}, \ldots , j_{N}} P (j_{1}, j_{2}, \ldots, j_{N}) = I_{\cH_{S}\otimes \cH_{A}}, \label{eq:proj}
\end{equation}
for every $i=1,2,\ldots, N$. Furthermore, it is easy to see that the operators $\{P_{i}(j_{i})\}$ constitute a set of pairwise orthogonal projectors.
\begin{eqnarray}
&& P_{i}(j_{i})P_{l}(k_{l}) \nonumber \\
&=& \sum_{ \{j_{m}\}\setminus j_{i}} P (j_{1}, \ldots, j_{N}) \sum_{\{k_{r}\} \setminus k_{l}} P (k_{1}, \ldots, k_{N}) \nonumber \\
&=& \sum_{ \{j_{m}\}\setminus j_{i},\{k_{r}\} \setminus k_{l} } \delta_{j_{1}k_{1}}\ldots \delta_{j_{N}k_{N}} \, P (j_{1}, \ldots, j_{N}) \nonumber \\
&=& \delta_{ij}\delta_{j_{i}k_{l}}P_{i}(j_{i}),
\end{eqnarray}
where we have used $\sum_{ \{j_{m}\}\setminus j_{i}}$ to denote the sum over all the elements of the set $\{ j_{1}, j_{2}, \ldots, j_{N}\}$, excluding $j_{i}$. 

In other words, for each $i \in [1,N]$, the set of projectors $\cP_{i} = \{P_{i}(j_{i}) \vert j_{i} = 1,2,\ldots, n_{i}\}$, constitutes a projective measurement on the extended space $\cH_{S}\otimes\cH_{A}$. The fact that the operators $\{P_{i}(j_{i})\}$ are pairwise orthogonal implies that the corresponding PVMs $\{\cP_{i} \, \vert \, i=1,2,\ldots, N\}$ are indeed compatible, as noted in Lemma~\ref{lem:pvm}. Finally, we observe that the compatible PVMs $\{\cP_{i}\}$ constructed as marginals of the PVM $\cP$ are indeed obtained as Naimark dilations of the POVMs $\{\cM_{i}\}$, as desired. 
\begin{eqnarray}
&& \tr_{\cH_{S}\otimes \cH_{A}}[ P_{i}(j_{i}) (\rho\otimes \sigma_{A})] \nonumber \\
&=& \sum_{ \{j_{m}\}\setminus j_{i}} \tr_{\cH_{S}\otimes \cH_{A}} [ P (j_{1}, j_{2}, \ldots, j_{N}) (\rho\otimes \sigma_{A}) ] \nonumber \\
&=& \sum_{ \{j_{m}\}\setminus j_{i}} \tr_{\cH_{S}} [ E(j_{1}, j_{2}, \ldots, j_{N}) \rho ] \nonumber \\
&=& \tr_{\cH_{S}} [ E_{i} (j_{i}) \rho ] . 
\end{eqnarray}

We now prove the converse. Suppose the set $\chi = \{\cM_{i}\}$ of POVMs is such that there exists a set $\Pi = \{\cP_{i}\}$ of compatible Naimark extensions acting on the extended  space $\cH_{S}\otimes \cH_{A}$. Compatibility of the PVMs $\{\cP_{i}\}$ implies that they all pairwise commute. In other words, if $\{P_{i}(j_{i})\}$ denotes the set of projectors associated with the Naimark extension $\cP_{i}$, these projectors satisfy, $\forall j_{i} \in \{1,2,\ldots,n_{i}\}$ and $ k_{l} \in \{1,2,\ldots,n_{l}\}$, 
\[ [P_{i}(j_{i}), P_{l}(k_{l}) ] = 0, \, \forall i, l \in [1,N] . \]  
We now need to check that the set $\chi$ is a set of compatible POVMs. 

We first note that compatibility of the PVMs $\Pi$ demands that all the projective measurements $\cP_{1}, \cP_{2}, \ldots, \cP_{N}$ act on the states of one and the same extended Hilbert space $\cH_{S}\otimes \cH_{A}$. Further, for each $i \in [1,N]$, there exists a fixed state $\sigma_{A}^{i} \in \cB(\cH_{A})$ such that, for all $j_{i} \in [1, n_{i}]$,
\[\tr_{\cH_{S}} [E_{i}(j_{i}) \rho] = \tr_{\cH_{S}\otimes \cH_{A}}[P_{i}(j_{i}) \left( \rho \otimes\sigma_{A}^{i}\right)], \, \forall \, \rho \in \cB(\cH_{S}). \]
It is important to note here the states $\{\sigma_{A}^{i}\}$ are, in general, distinct density operators on the ancilla system. 

The fact that the projective measurements $\{\cP_{i}\}$ are compatible, implies that there exists a joint projective measurement $\cP \equiv \{P(j_{1}, \ldots,j_{N})\}$, with projective elements satisfying,
\[\sum_{\{j_{m}\}\setminus j_{i}} P(j_{1}, j_{2}, \ldots, j_{N}) = P_{i}(j_{i}) ,  \]
for $j_{i}\in [1, n_{i}]$, for all $i \in [1,N]$. We then use the elements of this joint PVM on $\cH_{S}\otimes\cH_{A}$ to define a set of operators acting on the space $\cH_{S}$, as,
\begin{equation}
E^{(i)}(j_{1}, j_{2},\ldots, j_{N}) = \tr_{\cH_{A}} [P(j_{1}, \ldots, j_{N}) \sigma_{A}^{i}], \label{eq:joint}
\end{equation}
for $j_{l}\in [1, n_{l}]$ and for all $i \in [1,N]$. Clearly the operators $\{E^{(i)}(j_{1}, j_{2},\ldots, j_{N})\}$ are positive, since, 
\begin{eqnarray}
&& E^{(i)}(j_{1}, j_{2},\ldots, j_{N}) = \tr_{\cH_{A}} [P(j_{1}, \ldots, j_{N}) \sigma_{A}^{i}] \nonumber \\
&=& \tr_{\cH_{A}} [P(j_{1}, \ldots, j_{N}) (I_{\cH_{S}}\otimes\sigma_{A}^{i}) P(j_{1}, \ldots, j_{N})] \nonumber \\
&\geq & 0. \nonumber
\end{eqnarray}
Furthermore, the operators defined in Eq.~\ref{eq:joint} satisfy,


\begin{eqnarray}
&& \sum_{j_{1}, j_{2}, \ldots, j_N} E^{(i)} (j_{1},j_{2}, \dots, j_{N} )  \nonumber \\
&=& \sum_{j_{1}, j_{2}, \ldots, j_{N}}\tr_{\cH_{A}}\left[ (I_{\cH_{S}}\otimes\sigma_{A}^{i}) P ( j_{1}, j_{2}, \ldots, j_{N}) \right] = I_{\cH_{S}}, \nonumber
\end{eqnarray}
where the last equality follows from Eq.~\eqref{eq:proj}. Therefore, for each $i \in [1,N]$, $\cE^{(i)} = \{E^{(i)}( j_{1},j_{2}, \dots, j_{N})\}$ constitutes a valid POVM on the space $\cH_{S}$. 

For any state $\rho \in \cB(\cH_{S})$, the POVM elements $E^{(i)}( j_{1},j_{2}, \dots, j_{N})$ satisfy, 
\begin{eqnarray}
&& \tr_{\cH_{S}}[\rho E^{(i)}(j_{1},\ldots, j_{N} )] \nonumber \\
&=& \tr_{\cH_{S}}[\rho \, (\tr_{\cH_{A}} [\sigma_{A}^{i} P(j_{1}, \ldots, j_{N})]) \, ]\nonumber \\
&=& \tr_{\cH_{S}\otimes \cH_{A}} [ (\rho\otimes\sigma^{i}_{A}) P (j_{1}, \ldots, j_{N})] .\nonumber 
\end{eqnarray}
This implies that the marginals of the POVM elements defined in Eq.~\ref{eq:joint} satisfy,
\begin{eqnarray}
&& \tr_{\cH_{S}} \left[ \sum_{ \{j_{m}\}\setminus j_{l}} E^{(i)}(j_{1},\ldots, j_{N}) \rho \right] \nonumber \\
&=& \tr_{\cH_{S}\otimes \cH_{A}} \left[ (\rho\otimes\sigma^{i}_{A}) \sum_{ \{j_{m}\}\setminus j_{l}} P(j_{1}, \ldots, j_{N}) \right] \nonumber \\
&=& \tr_{\cH_{S}\otimes \cH_{A}} [ (\rho\otimes\sigma^{i}_{A}) P_{l}(j_{l}) ] \nonumber\\
&=& \tr_{\cH_{S}} [ E^{(i)}_{l}(j_{l}) \rho] , \forall \rho \in \cB(\cH_{S}) . \label{eq:marginal1}
\end{eqnarray}
Here we have defined the operators,
\begin{eqnarray}
E^{(i)}_{l}(j_{l}) &=& \tr_{\cH_{A}} [ P_{l}(j_{l}) (I_{\cH_{S}}\otimes \sigma^{i}_{A})], \nonumber \\
&& \; \forall \, i,l \in [1,N], \, j_{l} \in [1, n_{l}]. \label{eq:marginal2}
\end{eqnarray}
It is easy to check that the operators $E^{(i)}_{l}(j_{l}) \geq 0$, for all $i,l \in [1,N], j_{l} \in [1, n_{l}]$. Furthermore, 
\begin{eqnarray}
&& \sum_{j_{l}=1}^{n_{l}} E^{(i)}_{l}(j_{l}) \nonumber \\
&=& \tr_{\cH_{A}} [ \sum_{j_{l}=1}^{n_{l}} P_{l}(j_{l}) (I_{\cH_{S}}\otimes \sigma^{i}_{A})] \nonumber \\
&=& \tr_{\cH_{A}} [I_{\cH_{S}}\otimes \sigma^{i}_{A}] = I_{\cH_{S}} , \, \forall \, i,l\in [1,N].
\end{eqnarray} 
Thus, the set of positive operators $\{E^{(i)}_{l}(j_{l}) \vert j_{l} \in [1,n_{l}] \}$, constitutes a POVM $\cE^{(i)}_{j}$ on $\cH_{S}$, for each $i,l \in [1,N]$. Note that when $i=l$, 
\[ E^{(i)}_{i}(j_{i}) = \tr_{\cH_{A}} [\sigma_{A}^{i} P_{i}(j_{i})] = E_{i}(j_{i}),\]
and we recover the elements of the POVMs in the set $\chi \equiv \{\cM_{i} \}$. 
However, when $i \neq l$, the operators $E^{(i)}_{l}(j_{l})$ are not the same as the POVM elements $\{E_{i}(j_{i})\}$. 

We now consider two cases.
\begin{itemize}
\item[(i)] Suppose the ancilla states $\sigma_{A}^{i}$ are all the same, that is, $\sigma_{A}^{1}=\sigma_{A}^{2} = \ldots = \sigma_{A}^{N} = \sigma_{A}$(say). Then, it follows from the definition in Eq.~\ref{eq:marginal2} that, 
\[ E_{l}^{(1)}(j_{l}) = E_{l}^{(2)}(j_{l}) = \ldots = E_{l}^{(N)}(j_{l}). \] Therefore, $E_{l}^{(i)}(j_{l}) = E_{l}^{(l)}(j_{l})$, independent of $i$. As already noted above, $E_{l}^{(l)}(j_{l}) = E_{l}(j_{l})$ are indeed the POVM elements corresponding to the measurements in the set $\chi = \{\cM_{l} \}$. In other words, the POVM $\cE$ comprising the elements defined in Eq.~\ref{eq:joint} is indeed the joint measurement corresponding to the POVMs in the set $\chi = \{\cM_{i}\}$, thus showing that the set $\chi$ is indeed jointly measurable.
\item [(ii)] Suppose the ancilla states $\sigma_{A}^{i}$ are not the same, but distinct for each $i \in [1,N]$. In this case, it is not necessarily true that the POVMs $\{\cE_{i}\}$ are jointly measurable, although they have compatible Naimark extensions.
\end{itemize}
Thus we have shown that even when there exist of compatible Naimark extensions corresponding to a set of POVMs, the POVMs are compatible \emph{only if} these compatible Naimark extensions are obtained using a single ancilla system and identical density operators on this ancilla system.  
\end{proof}


\begin{remark}
It is worth noting here that related notion of \emph{joint dilatability} of  a set of POVMs was introduced in \cite{jointmeas_graph}. A set of POVMs $\{\cM_{i}, \, i=1,\ldots, N\}$ on a Hilbert space $\cH_{S}$ is said to be jointly dilatable, if there exists an ancillary space $\cH_{A}$ and a \emph{single} isometry $V : \cH_{S} \rightarrow \cH_{S} \otimes \cH_{A}$ and a \emph{single} PVM $\cP$ such that, the POVM effects satisfy $E_{i}(j_{i}) = \sum_{\{j_{k}\}\setminus j_{i}}V^{\dagger}P(j_{1}, j_{2}, \ldots, j_{N}) V $. From the definition, it follows that a set of POVMs is jointly measurable if and only if it is jointly dilatable.

We now show that our result in Theorem~\ref{thm:N_povm} provides a stronger, and arguably more useful, characterization of joint measurability than the observation in~\cite{jointmeas_graph}. To this end, we first recast the notion of joint dilatability in our formalism. Suppose $V$ be an isometry which extends a given POVM effect $E \in {\cal B}({\cal H}_S)$ to a projector $P \in {\cal B}({\cal H}_S \otimes {\cal H}_A)$ in such a way that,  $V^{\dagger}PV=E$. Then,
\begin{align}
\tr(\rho E)&=\tr(V^{\dagger}PV\rho)\nonumber\\
&=\tr(P(U\rho\otimes\sigma U^{\dagger}))\label{uni_delati}
\end{align}
for some unitary  $U\in\cB(\cH_S\otimes\cH_A)$ and ancilla state $\sigma\in\cH_A$. It is easy to see that 
\begin{equation}
\tr(V^{\dagger}PV\rho)=\tr(P(\rho\otimes\sigma)),
\end{equation}
if and only if $[P,U]=0$. From the proof of Theorem~\eqref{thm:N_povm}, it is clear that in our case, we restrict ourselves to those isometries for which the corresponding unitary commutes with the projectors on extended Hilbert space and obtain the if and only if result on joint measurability. The fact that we use this restricted set makes the result in Theorem~\eqref{thm:N_povm} a stronger statement. Experimentally, this means that to implement the joint POVM of set of compatible observable, we just have to add an ancilla and implement a suitable PVM on it; we do not need to implement any joint unitary before the implementation of that PVM.
\end{remark}

While we have used the tensor product form for the Naimark extensions in proving our result, a similar argument works for Naimark extensions which are constructed by using the direct sum. From the proof of the necessary part of Theorem~\eqref{thm:N_povm}, we know that commutativity of any kind of Naimark extension implies compatibility, and, from the proof of sufficiency we get that for all compatible set of observables we always get a commuting \emph{canonical} (tensor product) Naimark extension. So, we conclude that for any set of observables, if a commuting common direct sum Naimark extension exists, then a commuting common canonical Naimark extension also exists. Therefore in some cases we can restrict our attention to the canonical Naimark extensions. For example, only the canonical form of the Naimark extension has been considered in \eqref{uni_delati} above.

In what follows, we will discuss a few examples to highlight the central aspect of our result, namely, the connection between compatibility and the existence of a single common Naimark extension.

%

\subsection{Examples}\label{sec:example}

The existence of a \emph{single common} Naimark extension -- characterised by a specific ancilla Hilbert space and a fixed ancilla state -- that extends all the POVMs in the set $\cX$ to a set of commuting PVMs plays a crucial role in the proof of our necessary and sufficient condition for the compatibility of a set of POVMs. Here, we discuss three concrete examples of quantum measurements in $d=2$ to elucidate this aspect further, each individual example having its own significance.

%


\begin{example}[Commuting Naimark extensions for incompatible observables]
Consider the spin measurement along the $x$-direction of the Bloch sphere, characterized by the projectors $A_x=\{\ket{0}\bra{0},\ket{1}\bra{1}\}$ on a single-qubit space $\cH_{S}$. One choice of Naimark extension of this observable is $\cP_x = \{P_x, I-P_x \}$, where $P_x$ is a projector on a $4$-dimensional space $\cH_{S}\otimes\cH_{A}$, given by $P_{x}  =\ket{00}\bra{00}+\ket{+1}\bra{+1}$. This Naimark dilation can be obtained by setting the ancilla system to the fiducial state $|0\rangle\langle 0|$, as seen below. 
\[ \tr_{\cH_{S}\otimes\cH_{A}} [ (\rho\otimes\ket{0}\bra{0} ) P_x ] =\tr_{\cH_{S}} [ \rho \ket{0}\bra{0}].\]
However, if we set the ancilla state to $\ket{1}\bra{1}$, we see that,
\[ \tr_{\cH_{S}\otimes\cH_{A}} [ (\rho\otimes\ket{1}\bra{1}) P_x ] = \tr_{\cH_{S}} [\rho \ket{+}\bra{+}].\]
Thus, $\cP_x$ is also a Naimark extension of the spin measurement along the $y$-direction, characterized by the projectors $A_y= \{\ket{+}\bra{+},\ket{-}\bra{-}\}$ if we assume the ancilla state to be $\ket{1}\bra{1}$. Since $A_{x}$ and $A_{y}$ have a common Naimark extension, clearly they admit commuting Naimark extensions, although $A_x$ and $A_y$ are not compatible. Therefore it is clear that Naimark extensions using different ancillas cannot be used to characterize compatibility or incompatibility.
\end{example}

Our next example shows that if we use same ancilla states for all observables one can indeed guarantee commutativity from compatibility and vice-versa. Here the principal system is a two-dimensional Hilbert space, whereas the ancilla is eight dimensional. 

\begin{example}[Unsharp spin observables]
Consider the unsharp versions of the three observables $\sigma_x$, $\sigma_y$, $\sigma_z$ and with $\lambda$ denoting the unsharpness parameter. The POVM effects of these unsharp observables are,
\begin{eqnarray}
\Pi_{s}(1) &=& \lambda \ket{+}\bra{+} +(1-\lambda)\frac{I}{2}=\frac{\Id+\lambda \sigma_{s}}{2} \nonumber \\
\Pi_{s}(2) &=& \lambda \ket{-}\bra{-} +(1-\lambda)\frac{I}{2}=\frac{\Id-\lambda \sigma_{s}}{2},
\end{eqnarray}
where $s\in\{x,y,z\}$. 
\end{example}
We know that the POVMs $\cM_{s} = \{\Pi_{s}(1), \Pi_{s}(2)\}$ are compatible for $0\leq \lambda \leq \frac{1}{\sqrt{3}}$ \cite{Busch}. Furthermore, the joint POVM $\Pi=\{\Pi (i,j,k)\}$, $i,j,k \in \{1,2\}$ of these unsharp observables can be written down as,
\begin{eqnarray}
\Pi (i,j,k) &=& \frac{\Id+(-1)^{i+1}\lambda\sigma_x+(-1)^{j+1}\lambda\sigma_y+(-1)^{k+1}\lambda\sigma_z}{8}\nonumber\\
&=& \frac{\Id+\sqrt{3}\lambda (\vec{m}_{ijk}.\vec{\sigma})}{8},
\end{eqnarray}
where,
\begin{eqnarray}
 \vec{m}_{111} &=& \frac{1}{\sqrt{3}}(1,1,1) = -\vec{m}_{222} , \nonumber \\
 \vec{m}_{112} &=& \frac{1}{\sqrt{3}}(1,1,-1) = -\vec{m}_{221}, \nonumber \\
 \vec{m}_{121} &=& \frac{1}{\sqrt{3}}(1,-1,1) = -\vec{m}_{212}, \nonumber \\
 \vec{m}_{211} &=& \frac{1}{\sqrt{3}}(-1,1,1) = - \vec{m}_{122} . \nonumber 
 \end{eqnarray}
It is easy to check that $\sum_{ijk}\Pi (i, j, k)= I_{\cH_{S}} $. For all $i,j,k=\{1,2\}$, the operators $\{\Pi (i,j,k)\}$ are simultaneously positive for $0 \leq \lambda \leq \frac{1}{\sqrt{3}} $ and we recover the unsharp spin-$\frac{1}{2}$ observables as the marginals $\Pi_{x}(i)=\sum_{jk}\Pi (i,j,k)$, $\Pi_{y}(j)=\sum_{ik}\Pi_(i,j,k)$, and $\Pi_{z}(k)=\sum_{ij}\Pi (i,j,k)$. 

We now construct the Neimark extensions corresponding to these unsharp spin observables following the approach in~\cite{coherence_povm}. We start by observing that the square-root operators corresponding to the POVM elements $\Pi(i,j,k)$, can be written as,
%
%
%
%
%
%
%

\begin{eqnarray}
\sqrt{\Pi (i, j, k)} &=& \sqrt{\frac{1+\sqrt{3}\lambda}{8}}\ket{\vec{m}_{ijk},+}\bra{\vec{m}_{ijk},+} \nonumber\\
&+& \sqrt{\frac{1-\sqrt{3}\lambda}{8}}\ket{\vec{m}_{ijk},-}\bra{\vec{m}_{ijk},-}. 
\end{eqnarray}
Next we define $A_0=U\sqrt{\Pi (1,1,1)}$, $A_1=U\sqrt{\Pi (1,1,2)}$, $A_2=U\sqrt{\Pi (1,2,1)}$, $A_3=U\sqrt{\Pi (2,1,1)}$, $A_4=U\sqrt{\Pi(1,2,2)}$, $A_5=U\sqrt{\Pi(2,1,2)}$, $A_6=U\sqrt{\Pi(2,2,1)}$, $A_7=U\sqrt{\Pi(2,2,2)}$, where $U$ can be an arbitrary unitary matrix on the system Hilbert space $\cH_{S}$. Without loss of generality, we may take $U=I_{\cH_{S}}$.  Now we construct an isometry acting on the extended space $\cH_{S}\otimes \cH_{A}$, as described in~\cite{coherence_povm}, 
\begin{equation}
\tilde{V}=\sum_{m=0}^{7} (A_{m})_{\cH_{S}} \otimes (\ket{m}\bra{0}) _{\cH_{A}}.
\end{equation}
Note that $\cH_A$ is $8$-dimensional and $\{\ket{i}_{\cH_A}\}$ forms a basis for $\cH_A$. This isometry can be extended to a unitary $V$ on the Hilbert space $\cH_S \otimes \cH_A$ with the following form.
\begin{equation}
V=\sum_{m,a} (A_{m,a}) _{\cH_{S}}\otimes (\ket{m}\bra{a})_{\cH_{A}} . 
\end{equation}
with $A_{m,0}=A_m$. Unitarity implies that the operators $A_{m,a}$ have to satisfy,
\begin{equation}
\sum_{m=0}^{7} A^{\dagger}_{m,a}A_{m,b}=\delta_{ab} I_{\cH_{S}}, \; \sum_{a} A_{i,a}A^{\dagger}_{j,a}=\delta_{ij} I_{\cH_{S}}.
\end{equation}
The Naimark extension of the POVM $\Pi$ is characterized by the projectors,
\begin{equation}
P(i) = V^\dagger ( I\otimes\ket{i}\bra{i}) V, \, \forall \, i\in [0,7].
\end{equation}
Relabelling the projectors $P(0) = P(1,1,1)$, $P(1)=P(1,1,2)$, $P(2)=P(1,2,1)$, $P(3)=P(2,1,1)$, $P(4)=P(1,2,2)$, $P(5)=P(2,1,2)$, $P(6)=P(2,2,1)$, $P(7)=Pi(2,2,2)$, we have,
\begin{equation}
\pi P(i,j,k) \pi=\Pi (i,j,k) \ket{0}\bra{0}
\end{equation}
where $\pi = I_{\cH_{S}}\otimes(\ket{0}\bra{0})_{\cH_{A}}$. This in turn confirms that the projective measurement characterized by the set $\{P(i,j,k)\}$ does indeed correspond to a valid Naimark extension of the POVM with elements $\{\Pi(i,j,k)\}$. 
\begin{equation}
\tr [ (\rho\otimes\ket{0}\bra{0}) P(i,j,k) ] = \tr [\rho\Pi (i,j,k)]
\end{equation}
If we now consider the marginals of the projectors, we get,
\begin{eqnarray}
\tr [ (\rho\otimes\ket{0}\bra{0})\sum_{jk}P(i,j,k)] = \tr [\rho\Pi_{x}(i)], \nonumber \\
\tr[ (\rho\otimes\ket{0}\bra{0})\sum_{ik}P (i,j,k) ] = \tr [\rho\Pi_{y}(j) ], \nonumber \\
\tr [ (\rho\otimes\ket{0}\bra{0}) \sum_{ij}P(i,j,k)]= \tr [\rho\Pi_{z}(k)].
\end{eqnarray}
In other words, the set of projectors $\{P_{x}(i)=\sum_{jk}P(i,j,k) \, | \, i\in\{1,2\}\}$ correspond to the Naimark extension of $\{\Pi_{x}(i)\}$,  $\{P_{y}(j)=\sum_{ik}P(i,j,k)\}$ is the Naimark extension of $\{\Pi_{y}(j)\}$ and $\{P_{z}(k)=\sum_{ij}P(i,j,k)\}$ is the Naimark extension of $\{\Pi_{z}(k)\}$. As $\{P_{x}(i)\}$, $\{P_{y}(j)\}$ and $\{P_{z}(k)\}$ have a joint PVM $P(i,j,k)$, they are compatible and therefore pairwise commuting.

%

\begin{example}[Non-commuting Naimark extensions for compatible observables]\label{ex:non_com}
Let us consider the following two unsharp spin-$\frac{1}{2}$ observables $A=\{A(1)=\sigma_x(+,\lambda)=\frac{I+\lambda \sigma_x}{2},A(2)=\sigma_x(-,\lambda)=\frac{I-\lambda \sigma_x}{2}\}$ and $B=\{B(1)=\sigma_y(+,\lambda)=\frac{I +\lambda \sigma_y}{2}, B(2)=\sigma_y(-,\lambda)=\frac{I -\lambda \sigma_y}{2}\}$. We know that for $\lambda\leq \frac{1}{\sqrt{2}}$, those two pairs are compatible. So, theorem\eqref{thm:N_povm} states that there exists a commuting common Naimark extension. But do there exist common Naimark extensions of the observables $A$ and $B$ (for $\lambda \le 1/{\sqrt{2}}$) which do not commute? We show below that indeed there exist such Naimark extensions.
\end{example}

Consider the following Naimark extension of the effect $A(1)$ with respect to a two-dimensional ancilla Hilbert space, with  the ancilla state set to $\ket{0}\bra{0}$.
\begin{align}
P(1) &=\sigma_x(+,\lambda)\otimes \ket{0}\bra{0}+\frac{\sqrt{1-\lambda^2}}{2} I \otimes \ket{0}\bra{1}\nonumber\\
&\frac{\sqrt{1-\lambda^2}}{2} I \otimes \ket{1}\bra{0}+\sigma_x(-,\lambda)\otimes \ket{1}\bra{1}.
\end{align}
Note that we have used the Naimark construction described in Appendix~\ref{sec:naimark} with the unitary $U$ in Eq.~\eqref{Naimark} set to the identity operator on the system space. This guarantees that the operator $P(1)$ is indeed a projector on the system Hilbert space. The projector corresponding to the second effect $A(2)$ is obtained as $P(2) = I_{4\otimes 4} - P(1)$. 

Similarly, the Naimark extension of the effect $B(1)$ with respect to the same two dimensional ancilla state $\ket{0}\bra{0}$ can be obtained as,
\begin{align}
Q(1) &=\sigma_y(+,\lambda)\otimes \ket{0}\bra{0}+\frac{\sqrt{1-\lambda^2}}{2} I \otimes \ket{0}\bra{1}\nonumber\\
&+\frac{\sqrt{1-\lambda^2}}{2} I \otimes \ket{1}\bra{0}+\sigma_y(-,\lambda)\otimes \ket{1}\bra{1}.
\end{align}
The corresponding Naimark extension of $B(2)$ is given by, $Q(2) = I_{4 \times 4} - Q(1)$.

Now, one can easily check that,
\begin{eqnarray}
&& [P(1),Q(1)] \nonumber \\
&=& \frac{\lambda^2}{2}i\sigma_z\otimes\ket{0}\bra{0}+\frac{\lambda\sqrt{1-\lambda^2}}{2}(\sigma_x-\sigma_y)\otimes\ket{0}\bra{1}\nonumber\\
&+& \frac{\lambda\sqrt{1-\lambda^2}}{2}(\sigma_y-\sigma_x)\otimes\ket{1}\bra{0}+\frac{\lambda^2}{2}i\sigma_z\otimes\ket{1}\bra{1}\nonumber\\
&\Rightarrow & [P(1),Q(1)] \neq 0 ,
\end{eqnarray}
for all $\lambda$ satisfying $0 < \lambda \le 1$. 

It can be easily verified here that $Tr[{\rho}_SA(i)] = Tr[({\rho}_S \otimes |0{\rangle}_A{\langle}0|)P(i)]$ and $Tr[{\rho}_S B(i)] = Tr[({\rho}_S \otimes |0{\rangle}_A{\langle}0|)Q(i)]$ for $i = 1, 2$. Here the subscript $S$ refers to the system and $A$ refers to the ancilla.

Thus we see that, although the set $\chi = \{A, B\}$ of two dichotomic POVMs $A$ and $B$ for a spin-1/2 particle is compatible for $0 \leq \lambda \leq 1/{\sqrt{2}}$, one can still find out a set $\{{\cal P} = \{P(1), P(2)\}, {\cal Q} = \{Q(1), Q(2)\}\}$ of common but non-commuting Naimark extensions of the POVMs $A$ and $B$ in this range of $\lambda$.  On the other hand, by Theorem~\ref{thm:N_povm}, there will always exist a set of commuting but common Naimark extensions of these two compatible POVMs.

The aforesaid example signifies that for any given set of compatible POVMs for a system, apart from finding out common but commuting sets of Naimark extensions of these POVMs, one may also find out a common but non-commuting set of Naimark extensions -- although such a feature does not hold good in the case of any set of incompatible POVMs -- as, for any given set of incompatible POVMs, each set of its common Naimark extensions must be non-commuting, according to Theorem~\ref{thm:N_povm}.

Therefore, a pertinent question that may be raised at this point, is the following. Given a set of POVMs for one and the same system, can there exist at least one set (call it as a  `characteristic set') of Naimark extensions of these POVMs so that the incompatibility of such a set of Naimark extensions will guarantee the incompatibility of the original set of POVMs? A related issue, in case such set(s) of incompatibility-assuring  Naimark extension(s) exists (exist), would be to figure out the minimal dimension of the associated ancilla Hilbert space.  

\section{Applications}\label{sec:applications}

We now discuss two important applications of our result characterizing the incompatibility of a set of POVMs via their Naimark extensions. We first briefly discuss how this characterization can be used to quantify the incompatibility of a general set of quantum measurements, based on existing incompatibility measures for set of projection-valued measurements. Secondly, we obtain a simple sufficiency condition for the incompatibility of a pair of dichotomic observables, which can then be used to obtain the incompatibility region for a pair of unsharp spin-$\frac{1}{2}$ observables. 

\subsection{Quantifying incompatibility via the Naimark extension}\label{sec:measure}

While a plethora of measures exist that quantify the incompatibility of set of projection-valued measures, there remain several unresolved questions when it comes to  quantifying the incompatibility of set of POVMs. Much of the literature on the subject has focussed on a robustness-based appraoch to quantifying incompatibility of POVMs, which essentially involves quantifying the amount of noise that must be added to a pair of measurements to make them compatible~\cite{heinosaari15}.  

Recent studies show that the universal quantification of incompatibility of a set of quantum measurements through a scalar function is not very easy task. For example, unlike in the case of projective measurements, where mutually unbiased bases (MUBs) are known to correspond to the most incompatible measurements no matter what the metric of incompatibility is, in the case of POVMs different measures of incompatibility indicate different types of measurements to be most incompatible and therefore those measures can not be equivalent measures~\cite{Designolle}. In what follows, we outline an alternate approach to quantifying the incompatibility of a set of POVMs, based on the Naimark extension theorem. 
In this context, it is worth mentioning that a resource theory of incompatibility has been demonstrated in \cite{ Buscemi}.

Let $\chi = \{\cM_{i}\}$ be a set of POVMs and the set $\Pi ^{\Lambda}$ denote the set of projections obtained via a specific Naimark extension $\Lambda$ of each of the POVMs in the set $\chi$. Note that $\Lambda$ is characterized by the ancilla Hilbert space $\cH_A$ of dimension $d_{A}$ and a single fiducial ancilla state $\sigma$. Let  $f(\Pi^{\Lambda})$ be a function which measures the incompatibility of the set $\Pi^{\Lambda}$ such that $f(\Pi^{\Lambda})\geq 0$ for any such $\Pi^A$ and $f(\Pi^{\Lambda})=0$ iff $\Pi^{\Lambda}$ is compatible. Now, we can define an incompatibility measure $g(\chi)$ for the set of POVMs $\chi$ in terms of $f ( . )$ as,
\begin{equation}
g(\chi) = \min_{\Lambda: \Lambda \in \cN} f(\Pi^{\Lambda})
\end{equation}
where $\cN$ denotes the set of all possible Naimark extensions characterized by the ancilla system $\cH_{A}$ and the same ancilla state $\sigma_{A} \in \cB(\cH_{A})$. Theorem~\ref{thm:N_povm} clearly implies that $g(\chi)=0$ if and only if there exists an ancilla Hilbert space $\cH_A$ and a state $\sigma_{A}\in\cB(\cH_{A})$ such that the set $\Pi^{\Lambda}$ realised by the corresponding Naimark extension $\Lambda$ is compatible.  $g(\chi) > 0$, whenever such a joint Naimark extension does not exist, hence showing that $g (.)$ is indeed a faithful measure of incompatibility. 

The minimization over all Naimark extensions is indeed a difficult task, and we address possible ways by which this can be made easier in our concluding section.

\begin{example}
The measure of incompatibility for set of PVMs $\Pi$ defined in \cite{cloning-incompatibility}

\[
    \cQ(\Pi)=
    \begin{cases}
        0 & \text{if $\Pi$ is compatible}\\
        $a positive quantity$ & \text{otherwise.}
    \end{cases}
\]

So, this measure can be extended for the set of POVM $\chi$ as follows.

\begin{equation}
\cQ(\chi)=\sup_{\mu}\cQ(\mu(\chi))
\end{equation}

where $\mu(\chi)$ is a Naimark extension of set $\chi$ using one and the same ancilla state and the same ancilla system.

Clearly

\[
    \cQ(\chi)=
    \begin{cases}
        0 & \text{if $\chi$ is compatible}\\
        $a positive quantity$ & \text{otherwise.}
    \end{cases}
\]

\end{example}

\subsection{Compatibility of dichotomic measurements}\label{sec:dichotomic}

Recall that a \emph{dichotomic} measurement is a two-outcome POVM comprising of exactly two effects. An easily checkable sufficient condition for the compatibility of two dichotomic observables has been derived in~\cite{beneduci}, but the question of whether a similar sufficient condition can be obtained for different classes of observables is still an open problem. Indeed it turns out that this sufficiency condition is too restrictive and cannot for example identify the compatibility condition for a pair of unsharp spin-$\frac{1}{2}$ observables. 

Here we use Theorem~\ref{thm:N_povm} to obtain a more general sufficiency condition for the compatibilty of  a pair of dichotomic observables, compared to the condition derived in~\cite{beneduci} and show that our condition can capture the full compatibility region even for MUBs. For certain special cases, our condition reduces to the condition obtained in~\cite{beneduci}.

\begin{theorem}\label{thm:dichotomic}
Two dichotomic observables $\cA=\{A(1),A(2)\}$ and $\cB=\{B(1),B(2)\}$ acting on Hilbert space $\cH_{S}$ are compatible if there exists a unitary operator $W$ on $\cH_{S}$ satisfying the following conditions.

\begin{eqnarray}
&& [ A(1),B(1)] = Y(W X)- (WX)^{\dagger}Y, \nonumber \\
&& \{A(1),Y W\} - \{ X W^{\dagger},B(1) \}W = YW - X , \nonumber \\
&& \left[ WA(1)W^{\dagger},B(1) \right] + (W X)Y-Y(W X)^{\dagger} = 0 , \label{eq:W_condition}
\end{eqnarray}
where $X=\sqrt{A(1)A(2)}$ and $Y=\sqrt{B(1)B(2)}$.
\end{theorem}

\begin{proof}
Let $\cP=\{P(1), P(2)\}$ be a Naimark extension of the dichotomic observable $\cA=\{A(1), A(2)\}$, constructed via the prescription given in Appendix~\ref{sec:naimark}. The projectors $P(1)$ and $P(2)$ are given in terms of the POVM effects $A(1)$ and $A(2) = I_{\cH_{S}} - A(1)$ as,

\begin{eqnarray}
P(1) &=&\begin{bmatrix}
A(1) & -\sqrt{A(1)(I-A(1))}U^{\dagger}\\
-U\sqrt{A(1)(I-A(1))} & U(I-A(1))U^{\dagger}\\
\end{bmatrix} , \nonumber \\
P(2) &=& \begin{bmatrix}
A(2) &\sqrt{A(2)(I -A(2))}U^{\dagger}\\
U\sqrt{A(2)(I - A(2))} & U(I-A(2))U^{\dagger}\\
\end{bmatrix} , \label{eq:A_proj}
\end{eqnarray}
where, $U$ is an arbitrary unitary matrix on the system space $\cH_{S}$. Clearly for $U=I_{\cH_{S}}$ or $U=-I_{\cH_{S}}$ these projectors result in the same Neimark extension used in~\cite{beneduci}.

It can be checked easily the set $\cP=\{P(1),P(2)\}$ does indeed constitute a PVM. Furthermore, for $i=\{1,2\}$, and for all $\rho \in \cB(\cH_{S})$,
\begin{equation}
\tr_{\cH_{S}} [ \rho A(i) ] =\tr_{\cH_{S}\otimes\cH_{A}} [ (\rho\otimes\ket{0}\bra{0}) P(i) ],
\end{equation}
where $|0\rangle\langle 0| \in \cB(\cH_{A})$ is a fiducial state of the ancilla system $\cH_{A}$. 

The Naimark extension corresponding to the second observable $\cB=\{B(1), B(2)\}$ is characterized by the projectors $\{Q(1), Q(2)\}$, given by,
\begin{eqnarray}
Q(1) &=& \begin{bmatrix}
B(1) & -\sqrt{B(1)(I-B(1))}V^{\dagger}\\
-V\sqrt{B(1)(I-B(1))}& V(I -B(1))V^{\dagger}\\
\end{bmatrix}, \nonumber \\
Q(2) &=& \begin{bmatrix}
B(2) & \sqrt{B(2)(I-B(2))}V^{\dagger}\\
V\sqrt{B(2)(I-B(2))} & V(I-B(2))V^{\dagger}\\
\end{bmatrix} , \label{eq:B_proj}
\end{eqnarray}
where $V$ is an arbitrary unitary matrix on $\cH_{S}$.

To check if the PVMs $\cP$ and $\cQ$ are compatible, we only need to check whether $[P(1),Q(1)]=0$, for this implies that all the other commutators also vanish, as shown below.
\begin{eqnarray}
[P(1),Q(2)] &=& \left[ P(1),(I-Q(1)) \right] = 0, \nonumber \\  
\left[ P(2),Q(1) \right] &=& \left[ (I-P(1)),Q(1) \right] = 0, \nonumber \\
\left[ P(2),Q(2) \right] &=& \left[ (I-P(1)),(I - Q(1)) \right] = 0. 
\end{eqnarray}

Let, $X=\sqrt{A(1)A(2)}$ and $Y=\sqrt{B(1)B(2)}$. Then, the projectors $P(1)$ and $Q(1)$ are given by,
\begin{eqnarray}
P(1) &=& \begin{bmatrix}
A(1) & -XU^{\dagger}\\
-UX & U(I-A(1))U^{\dagger}\\
\end{bmatrix}, \nonumber \\
Q(1) &=& \begin{bmatrix}
B(1) &-YV^{\dagger}\\
-VY & V(I-B(1))V^{\dagger}\\
\end{bmatrix}. \label{eq:naimark_form}
\end{eqnarray}

Since $[A(1) , I_{\cH_{S}} -A(1)]=0$ and $A(1)+A(2)=I_{\cH_{S}}$, the effects $A(1)$ and $A(2)$ share a common set of eigenstates. Therefore, $\sqrt{A(1)A(2)}=\sqrt{A(2)A(1)}=\sqrt{A(1)}\sqrt{A(2)}=\sqrt{A(1)(I-A(1))}=\sqrt{A(2)(I-A(2))}$. These relations are very crucial in the following calculations. In particular, the condition $[P(1),Q(1)]=0$ now reduces to,
\begin{eqnarray}
A(1)B(1) &+& XU^{\dagger}VY = B(1)A(1) + YV^{\dagger}UX \nonumber \\
A(1)YV^{\dagger} &+& XU^{\dagger}VB(2)V^{\dagger} \nonumber \\
&=& B(1) XU^{\dagger}+YV^{\dagger}UA(2)U^{\dagger}, \nonumber \\
UXYV^{\dagger} &+& UA(2)U^{\dagger}VB(2)V^{\dagger} \nonumber \\
&=& VYXU^{\dagger} + VB(2)V^{\dagger}UA(2)U^{\dagger}. \label{eq:cond1}
\end{eqnarray}


When both unitaries $U$ and $V$ are trivial, that is, when $U=V$, the set of equations in Eq.~\eqref{eq:cond1} are not independent and instead reduce to a set of two equations, namely,
\begin{eqnarray}
A(1)B(1)+XY &=& B(1)A(1) + YX, \nonumber \\
A(1)Y+XB(2) &=& B(1)X + YA(2), \label{eq:beneduci}
\end{eqnarray}
which are infact the conditions obtained in~\cite{beneduci}. 

Now defining the unitary matrix $W$ as $V^{\dagger}U=W$, the conditions in Eq.~\eqref{eq:cond1} can be rewritten as,
\begin{eqnarray}
[A(1),B(1)] &=& Y(W X)-(WX)^{\dagger}Y, \nonumber \\
\{A(1),YW\} &-& \{X W^{\dagger},B(1)\} W =Y W-X, \nonumber \\
\left[ WA(1) W^{\dagger}, B(1) \right] &+& ( WX)Y-Y(WX)^{\dagger} = 0,
\end{eqnarray}
thus leading to the desired conditions for compatibility. 
\end{proof}

\subsection{Unsharp spin-$1/2$ observables}\label{sec:unsharp}

In this section we will show how the conditions proved in Theorem~\ref{thm:dichotomic} are helpful to calculate the compatibility region for a pair of unsharp spin-half observables.

\begin{theorem}\label{thm:spin}
Two spin unsharp spin half observables along the directions $\hat{n_1}$ and $\hat{n_2}$ with unsharpness parameter $\lambda_1$ and $\lambda_2$ respectively, are compatible if there exists a unitary matrix $W$ satisfying,
\begin{eqnarray}
&& \lambda_1\lambda_2[\hat{n}_1.\vec{\sigma},\hat{n}_2.\vec{\sigma}] \, = \, \sqrt{1-\lambda_1^2}\sqrt{1-\lambda_2^2}[W -W^{\dagger}], \nonumber \\
&& \lambda_1\sqrt{1-\lambda_2^2} \left\{\frac{\hat{n}_1.\vec{\sigma}}{2}, W\right\} \nonumber \\
&& \; \quad  - \; \lambda_2\sqrt{1-\lambda_1^2} W^{\dagger}\left\{\frac{\hat{n}_2.\vec{\sigma}}{2}, W \right\}=0,  \nonumber \\
&& \lambda_1\lambda_2[W\hat{n}_1.\vec{\sigma}W^{\dagger},\hat{n}_2.\vec{\sigma}] \nonumber \\
&& \; \quad + \; \sqrt{(1-\lambda_1^2)(1-\lambda_2^2)}[ W- W^{\dagger}] = 0.\label{eq:spin-compatibility condition}
\end{eqnarray}
\end{theorem}

\begin{proof}
Consider two unsharp spin half observables $\cA=\{A(1), I_{\cH_{S}} - A(1)\}$, $\cB = \{ B(1), I_{\cH_{S}} -B(1) \}$ with $A(1) = \frac{ I_{\cH_{S}} +\lambda_1\hat{n}_1.\vec{\sigma}}{2}$
and $B(1) = \frac{ I_{\cH_{S}} + \lambda_2\hat{n}_2.\vec{\sigma}}{2}$. Therefore, the operators $X$ and $Y$ defined in Theorem~\ref{thm:dichotomic} are in this case given by,
$X=\sqrt{A(1)(I-A(1)}=\frac{\sqrt{1-\lambda_1^2}}{2} I$ and
$Y=\sqrt{B(1)(I-B(1))}=\frac{\sqrt{1-\lambda_2^2}}{2} I$.

So, the set of equations in Eq.~\eqref{eq:W_condition} reduces to,
\begin{eqnarray}
&& \lambda_1\lambda_2[\hat{n}_1.\vec{\sigma},\hat{n}_2.\vec{\sigma}] \, = \, \sqrt{1-\lambda_1^2}\sqrt{1-\lambda_2}[W-W^{\dagger}], \nonumber \\
&& \lambda_1\sqrt{1-\lambda_2^2}\left\{\frac{\hat{n}_1.\vec{\sigma}}{2}, W \right\} \nonumber \\
&& \; \quad - \; \lambda_2\sqrt{1-\lambda_1^2}W^{\dagger}\left\{ \frac{\hat{n}_2.\vec{\sigma}}{2}, W\right\} =0, \nonumber \\
&& \lambda_1\lambda_2[ W\hat{n}_1.\vec{\sigma}W^{\dagger},\hat{n}_2.\vec{\sigma}] \nonumber \\
&& \; \quad + \; \sqrt{(1-\lambda_1^2)(1-\lambda_2^2)}[W-W^{\dagger}] = 0, \label{eq:spin_compat}
\end{eqnarray}
as desired. 
\end{proof}

\subsubsection{Compatibility region for $\sigma_x$ and $\sigma_y$}

Here we show that incompatibility region of the unsharp observables corresponding to $\sigma_x$ and $\sigma_y$ can be obtained using above condition. Interestingly, for $\hat{n_1}=\{1,0,0\}$ and $\hat{n_2}=\{0,1,0\}$ and for $W=e^{i\theta}\sigma_z$ all the above conditions is satisfied for,

\begin{align}
\lambda_1^2+\lambda_2^2&=1-\lambda_1^2\lambda_2^2\cot^2\theta.\label{eq:lambda}
\end{align}

Now, the LHS of Eq.~\eqref{eq:lambda} is positive. Hence the restriction on $\cot\theta$ for the existence of a valid solution for the above equation is,
\begin{equation}
\cot^2\theta\leq \frac{1}{\lambda_1^2\lambda_2^2}.
\end{equation}

Therefore by changing $\theta$ we can make the RHS of the above equation vary from $0$ to $1$. Hence the unsharp version of $\sigma_x$ and unsharp version of $\sigma_y$ are compatible for the parameter region given by,
\begin{equation}
\lambda_1^2+\lambda_2^2\leq 1.
\end{equation}

This is indeed the well known compatibility region for the unsharp versions of $\sigma_x$ and $\sigma_y$~\cite{Heinosaari_incompatibility_region,
Busch_incompatibility_region}. We expect that our condition would similarly enable us to characterize the compatibility region for sets of unsharp or noisy MUBs~\cite{designolle_mub}. It is possible that the construction used in Theorem~\ref{thm:spin} can capture the full compatibility region for any pair of spin $\frac{1}{2}$ observables.

Also, using $W=V^{\dagger}U = e^{i\theta}\sigma_z$ in the Naimark extensions of unsharp versions of $\sigma_x$ and $\sigma_y$ of the form given in \eqref{eq:naimark_form}, and thereafter taking the Jordon product and reducing to the system Hilbert space, one easily gets the joint observable of these two unsharp observables. Therefore, the condition in Eq.~\eqref{eq:spin-compatibility condition}is also very useful to get the joint observable easily. Moreover it is clear that, if for some cases more than one unitary exists then there exist many joint observables, one corresponding to each unitary.

%


Finally, we discuss the well known case of an incompatible  set of three observables in a two-dimensional Hilbert space, such that every pair of observables in the set is compatible~\cite{heinosaari08, jointmeas_graph,Liang,Kunjwal_Sibasish}. 

\begin{example}[Pairwise compatible, but triple-wise incompatible]
 Consider a set of three dichotomic qubit observables $\cA_1=\{A_1(1),A_1(2)\}$,$\cA_2=\{A_2(1),A_2(2)\}$, and $\cA_3=\{A_3(1),A_3(2)\}$ such that $A_1(1)=\frac{\Id+\lambda\sigma_x}{2}$,  $A_2(1)=\frac{\Id+\lambda\sigma_y}{2}$, and $A_3(1)=\frac{\Id+\lambda\sigma_z}{2}$ where $\lambda=\frac{1}{\sqrt{2}}$. This set is not compatible, since $\lambda>\frac{1}{\sqrt{3}}$. Therefore, theorem \eqref{thm:N_povm} says that compatible common Naimark extensions do not exist. Now consider the Naimark extensions $\cQ_1$, $\cQ_2$, and $\cQ_3$ of the observables $\cA_1$, $\cA_2$, and $\cA_3$ respectively, with respect to the two-dimensional ancilla state $\ket{0}\bra{0}$, such that, 
 \begin{align}
& Q_1(1)=\begin{bmatrix}
 A_1(1)&-\sqrt{A_1(1)A_1(2)} \\
-\sqrt{A_1(1)A_1(2)}&A_1(2) \\
 \end{bmatrix}\\
 &Q_2(1)=\begin{bmatrix}
 A_2(1)&-i\sqrt{A_2(1)A_2(2)}\sigma_z \\
i\sigma_z\sqrt{A_2(1)A_2(2)}&A_2(2) \\
 \end{bmatrix}\\ 
 &Q_3(1)=\begin{bmatrix}
 A_3(1)&i\sqrt{A_3(1)A_3(2)}\sigma_y \\
-i\sigma_y\sqrt{A_3(1)A_3(2)}&A_3(2) \\
 \end{bmatrix}.
 \end{align}
 Now, from theorem \eqref{thm:spin} one can easily check that $[Q_1(1),Q_2(1)]=0$, $[Q_1(1),Q_3(1)]=0$, but, $[Q_2(1),Q_3(1)]\neq 0$. Thus $\cQ_{1}, \cQ_{2}, \cQ_{3}$ constitute a set of common commuting Naimark extensions for the pairs $\{\cA_{1}, \cA_{2}\}$ and $\{\cA_{1}, \cA_{3}\}$, but not for the pair $\{\cA_{2}, \cA_{3}\}$.
 
 Now, consider another set of  common Naimark extensions $\cQ_2^{\prime}$, and $\cQ_3^{\prime}$ with respect to the two-dimensional ancilla state $\ket{0}\bra{0}$ for the observables $\cA_2$, and $\cA_3$ respectively, such that,
 \begin{align}
 &Q_2^{\prime}(1)=\begin{bmatrix}
 A_2(1)&-\sqrt{A_2(1)A_2(2)} \\
-\sqrt{A_2(1)A_2(2)}&A_2(2) \\
 \end{bmatrix}\\ 
 &Q_3^{\prime}(1)=\begin{bmatrix}
 A_3(1)&-i\sqrt{A_3(1)A_3(2)}\sigma_x \\
i\sigma_x\sqrt{A_3(1)A_3(2)}&A_3(2) \\
 \end{bmatrix}.
 \end{align}
 
 Again, from theorem \eqref{thm:spin} one can easily check that $[Q_2^{\prime}(1),Q_3^{\prime}(1)]=0$. Thus we see that it is indeed possible to obtain (distinct) common Naimark extensions for each pair of observables in the set $\{\cA_{1}, \cA_{2}, \cA_{3}\}$, showing that the set is indeed pairwise compatible.
\end{example}

\section{Conclusions}\label{sec:concl}

We have shown how the compatibility of a set of quantum measurements can be characterized via their Naimark extensions. Specifically, we prove that a set of POVMs is compatible if and only if they can be extended to a set of commuting projective measurements, via Naimark extensions constructed using the same ancilla Hilbert space and identical ancilla states. 

Our result opens up new avenues of research into the incompatibility of POVMs, a question of fundamental importance in the context of both quantum foundations as well as quantum information processing. Firstly, it provides a more physical and constructive approach to quantifying incompatibility of POVMs, in terms of unitaries and projective measurements on an extended space. Secondly, our work provides an alternate route to quantifying the incompatibility of a set of quantum measurements, based on existing measures of incompatibility that hold for projective measurements. Finally, in the case of dichotomic observables, our results enable us to obtain a simple and checkable sufficiency condition for the incompatibility of a pair of dichotomic observables in any dimension. This condition can then be used to characterize the compatibility regions for certain important classes of measurements, including unsharp qubit measurements. 

Given that the conditions obtained in Eq.~\eqref{eq:spin-compatibility condition} can capture full compatibility of any two unsharp spin half observables, we are tempted to conjecture that any two unsharp spin observables are compatible if and only if they satisfy Eq.~\eqref{eq:spin-compatibility condition}. Furthermore, using the form of the unitary $W$ that satisfies the compatibility conditions, one has an concrete prescription to construct the joint observable. We know that traditional calculation of the compatibility region and joint observable is difficult even for unsharp MUBs in higher dimensions~\cite{Designolle}. In our approach, once the suitable Naimark extension is constructed, further calculations are straightforward. So, in effect, we have a simpler procedure to obtain both of these. 



An important avenue for future research is the question raised at the end of Example~\ref{ex:non_com}, namely, whether it is possible to identify a \emph{characteristic set} of common Naimark extensions of a given set of POVMs such that the commutativity (non-commutativity) of the Naimark extensions in the set will guarantee the  compatibility (incompatibility) of the POVMs in the set. From the perspective of quantifying incompatibility of any given set $\chi$ of POVMs -- expressed in  terms of the quantity ${\cal Q}(\chi)$ -- identifying such a characteristic set of Naimark extensions is of fundamental importance, at least to get an estimate of the incompatibility of the set $\chi$.

\section{Acknowledgement}
We would like to thank Dr. Ravi Kunjwal and Dr. M. D. Srinivas for their valuable comments and suggestions.

\appendix
\section{Construction of Naimark extension}\label{sec:naimark}

The Naimark extension associated with a POVM is not unique and there are several approaches to obtaining a Naimark extension for a given POVM. In our work we make use of the tensor product construction described in~\cite{coherence_povm}. We provide the details of this Naimark construction here, for the specific case of dichotomic observables, for completeness. 

Consider an $N$-outcome observable $\cM$ acting on a $d$-dimensional space $\cH_{S}$ with effects $\{E(i) \, | \, I_{\cH_{S}} \geq E(i)\geq 0, E(i)\in\cB(\cH_S)\forall \{1,\ldots,N\}\}$ satisfying $\sum_i E(i)= I_{\cH_{S}}$. 
%
%
%
Following~\cite{coherence_povm}, the Naimark extension of the POVM $\cM$ can then be constructed via the projectors
\begin{equation}
P(i)=\sum_{a,b=0}^{N-1} (A^{\dagger}_{i,a}A_{i,b})_{\cH_{S}} \otimes(\ket{a}\bra{b})_{\cH_{A}}, \label{eq:proj}
\end{equation}
for all $i \in [1,N]$, where $A_{i,0}=U_i\sqrt{E(i)}$. For the case of dichotomic observables ($N=2$) the conditions on the operators $A_{i,0}$ reduces to the following:

\begin{eqnarray}
A^{\dagger}_{1,0}A_{1,0}+A^{\dagger}_{2,0}A_{2,0} &=& I_{\cH_{S}} \nonumber \\
A^{\dagger}_{1,0}A_{1,1}+A^{\dagger}_{2,0}A_{2,1}&=& 0  \nonumber \\
A^{\dagger}_{1,1}A_{1,0}+A^{\dagger}_{2,1}A_{2,0}&=& 0 \nonumber \\
A^{\dagger}_{1,1}A_{1,1}+A^{\dagger}_{2,1}A_{2,1} &=& I _{\cH_{S}} \nonumber \\
A_{1,0}A^{\dagger}_{1,0}+A_{1,1}A^{\dagger}_{1,1} &=& I_{\cH_{S}} \nonumber \\
A_{1,0}A^{\dagger}_{2,0}+A_{1,1}A^{\dagger}_{2,1} &=& 0 \nonumber \\
A_{2,0}A^{\dagger}_{1,0}+A_{2,1}A^{\dagger}_{1,1}&=& 0 \nonumber \\
A_{2,0}A^{\dagger}_{2,0}+A_{2,1}A^{\dagger}_{2,1} &=& I_{\cH_{S}} \label{eq:A_matrix}
\end{eqnarray}

Here, $A_{1,0}=U_1\sqrt{E(1)}$ and $A_{2,0}=U_2\sqrt{E(2)}$, where $U_{1}$ and $U_{2}$ are arbitrary unitary operators on the system Hilbert space $\cH_{S}$.  Thus, the problem of finding Naimark extension is now reduced to the problem of finding operators $A_{1,1}$ and $A_{2,1}$ simultaneously satisfying the set of equations in Eq.~\eqref{eq:A_matrix}. For simplicity we may take $U_0=U_1=U$. Now, suppose the POVM effects have a spectral decomposition given by,
\begin{equation}
E(i) = \sum_{j=1}^{d}\lambda_{ij}\ket{j}\bra{j},
\end{equation}
then ,
\begin{align}
\sqrt{UE(i)U^{\dagger}}&=\sqrt{\sum_{j}\lambda_{ij}U\ket{j}\bra{j}U^{\dagger}}\nonumber\\
&=\sum_{j}\sqrt{\lambda_{ij}}\ket{Uj}\bra{Uj}\nonumber\\
&=\sum_{j}\sqrt{\lambda_{ij}}U\ket{j}\bra{j}U^{\dagger}\nonumber\\
&=U\sqrt{E(i)}U^{\dagger}.
\end{align}

So, one choice of solution for the operator $A_{2,1}$ in Eq.~\eqref{eq:A_matrix} is,
\begin{align}
A_{2,1}&=\sqrt{I -UE(2)U^{\dagger}}\nonumber\\
&=U\sqrt{E(1)}U^{\dagger}
\end{align}

Corresponding to this operator $A_{2,1}$, one possible solution for $A_{1,1}$ in Eq.~\eqref{eq:A_matrix} is,
\begin{align}
A_{1,1}&=\sqrt{I -UE(1)U^{\dagger}}\nonumber\\
&=U\sqrt{E(2)}U^{\dagger}
\end{align}

It can be easily checked that for above solution for $A_{1,1}$ and $A_{2,1}$ all other equations are satisfied.

So, from equation\eqref{eq:proj}, one possible set of projectors $\{P(1), P(2)\}$ for the Naimark extension of $\cM$ is given by,
\begin{align}\label{Naimark}
P(1)&=\begin{bmatrix}
A^{\dagger}_{1,0}A_{1,0}&A^{\dagger}_{1,0}A_{1,1}\\
A^{\dagger}_{1,1}A_{1,0}&A^{\dagger}_{1,1}A_{1,1}
\end{bmatrix}\nonumber\\
&=\begin{bmatrix}
E(1)&-\sqrt{E(1)E(2)}U^{\dagger}\\
-U\sqrt{E(1)E(2)}&UE(2)U^{\dagger}
\end{bmatrix}
\end{align}

and

\begin{align}
P(2)&=\begin{bmatrix}
A^{\dagger}_{2,0}A_{2,0}&A^{\dagger}_{2,0}A_{2,1}\\
A^{\dagger}_{2,1}A_{2,0}&A^{\dagger}_{2,1}A_{2,1}
\end{bmatrix}\nonumber\\
&=\begin{bmatrix}
E(2)&\sqrt{E(1)E(2)}U^{\dagger}\\
U\sqrt{E(1)E(2)}&UE(1)U^{\dagger}
\end{bmatrix} .\label{Naimark-2}
\end{align}

\end{document}